\documentclass[12pt,twoside]{article}
\usepackage{latexsym, amsthm, amsmath, bbm}
\usepackage{amssymb}
\usepackage{amsfonts}
\usepackage{amstext}
\usepackage{multicol}
\usepackage{tabularx}
\usepackage{hyperref}
\usepackage{url}
\usepackage{appendix}
\usepackage{graphicx}
\usepackage{color}
\usepackage{array,multirow}
\newtheorem{Theorem}{Theorem}

\newtheorem{Proposition}{Proposition}

\newtheorem{Example}{Example}
\newtheorem{Corollary}{Corollary}

\newtheorem{Claim}{Claim}

\theoremstyle{remark}
\newtheorem{Remark}{Remark}

\newcolumntype{L}[1]{>{\raggedright\arraybackslash}p{#1}}
\newcolumntype{C}[1]{>{\centering\arraybackslash}p{#1}}
\newcolumntype{R}[1]{>{\raggedleft\arraybackslash}p{#1}}
\newcommand{\R}{{\mathbb{R}}}

\newcommand{\Z}{{\mathbb{Z}}}

\newcommand{\D}{{\mathcal{D}}}
\title{The Efficiency Gap, Voter Turnout, and the Efficiency Principle}

\author{
Ellen Veomett \\
\small Department of Mathematics and Computer Science \\[-0.8ex]
\small Saint Mary's College of California \\[-0.8ex]
\small Moraga, California, U.S.A \\
\small \tt erv2@stmarys-ca.edu
}

\begin{document}

\maketitle

\begin{abstract}
  Recently, scholars from law and political science have introduced  metrics which use only election outcomes (and not district geometry) to assess the presence of partisan gerrymandering.   The most high-profile example of such a tool is the \emph{efficiency gap}.   Some scholars have suggested that such tools   should be sensitive enough to alert us when two election outcomes have the same percentage of votes going to political party $A$, but one of the two awards party $A$ more  seats.  When a metric is able to distinguish election outcomes in this way, that metric is said to satisfy the \emph{efficiency principle}.  

In this article, we show that the efficiency gap fails to satisfy the efficiency principle.   We show precisely how the efficiency principle breaks down in the presence of unequal voter turnout.  To do this, we first present a construction that, given any rationals $1/4< V<3/4$ and $0<S<1$, constructs an election outcome with vote share $V$, seat share $S$, and EG = 0.   (For instance, one party can get 26\% of the vote and anywhere from 1\% to 99\% of the seats while the efficiency gap remains zero.)   Then, for any election with vote share  $1/4<V<3/4$, seat share $S$, and EG= 0, we express  the ratio $\rho$ of average turnout in districts  party $A$ lost to average turnout in districts  party $A$ won as a function in only $V$ and $S$.  It is well known that when all districts have equal turnout, EG can be expressed as a simple formula in $V$ and $S$;  we express the efficiency gap of \emph{any} election as an equation only in $V, S,$ and $\rho$.  We also report on the values of $\rho$ that can be observed in actual elections.
\end{abstract}

\section{Introduction}

Gerrymandering is an issue that has had a long history in the American democracy.  The term ``gerrymander'' started with the 1812 redistricting of Massachusetts.  In that year, Governor Elbridge Gerry participated in a redistricting that successfully kept Governor Gerry's Democratic-Republican party in power in the state senate.  

In the 1812 redistricting of Massachusetts, the telltale sign that the redistricting effort was intended to benefit one political party was the bizarre salamander shape of one of the resulting districts (hence the name ``Gerry-mander'').  While the irregular shape of districts has long been the focus of claims that districts have been gerrymandered, deciding whether a shape is reasonable or not is surprisingly difficult.  For example, one could consider a district's shape to be a problem if its boundary is too long, given the contained area.  But upon consideration of a state with a long coastline, or even the fact that boundaries can become very irregular simply upon ``zooming in'' with a digital map,  this boundary to area ratio reveals its weakness as a tool to measure unfairness, as M. Duchin and B. Tenner recently discuss in their preprint \cite{DuchinTenner}.  Indeed, this issue was mentioned in Schwartzberg's 1966 paper, when he emphasized using ``gross perimeter,'' which is essentially a smoothing of the actual perimeter \cite{Schwartzberg}.

Another natural way to measure the irregularity of a district's shape might be to compare how it differs from a circle, square, or other desired shape.  Or one could describe how far a district's shape is from being convex.  But Young's survey \cite{Young} easily illustrates the fact that comparing district shapes to fixed shapes leads to shapes which are visual oddities being scored as better than shapes like triangles or rectangles.  And measuring how far a shape is from being convex has its own shortcomings, as referenced in \cite{Lunday}.  

But perhaps even more important than the complication in measuring the oddities of a district's shape is the effect of modern technology on mapmaking.    Modern mapmakers can more readily make districting plans which satisfy reasonable shape requirements but produce a range of partisan results. And thus a mapmaker can easily choose from among those many maps the one which satisfies her political agenda.

Thus some have searched for a non-geometric tool that can measure how fair or unfair a redistricting plan is in terms of its partisan effects. But in order to have a useful tool, we must know exactly what that tool does and does not do. In this article we explore the tool called the efficiency gap, and we study its behavior with respect to a criterion for gerrymandering metrics called the efficiency principle. In Section \ref{EPSection} we define and discuss the efficiency principle, dividing it into statements  EP1 and EP2.  We define the efficiency gap in Section \ref{CalculatingEGSection} and explore how the efficiency gap fails to satisfy EP1 in Section \ref{EGnoEPSection}  and EP2 in Section \ref{EGnoPart2EP}.  We give a  construction in section  \ref{OneFourthSection} which shows shows that elections with EG=0 may still have a wild disproportion of seats to votes.   This construction leads us to study turnout ratios in Section \ref{TurnoutSection}, and we discuss the implications of those ratios in Section \ref{FixingTSection}.   We give some brief final comments in Section \ref{CommentsSection}.

\section{The Efficiency Principle}\label{EPSection}

We now set geometry aside and discuss tools designed to measure a partisan gerrymander based on election outcomes.   In \cite{IntroEfficiency}, McGhee introduced a property he called the ``efficiency principle,'' and Stephanopoulos and McGhee argued in \cite{MeasureMetricDebate} that this principle is necessary for any such tool.  Stephanopoulos and McGhee have a nice description of the efficiency principle in  \cite{MeasureMetricDebate}: 

\begin{quotation}
[The efficiency principle]  states that a measure of partisan gerrymandering ``must indicate a greater advantage for (against) a party when the seat share for that party increases (decreases) without any corresponding increase (decrease) in its vote share'' \cite{IntroEfficiency}. The principle would be violated, for example, if a party received 55\% of the vote and 55\% of the seats in one election, and 55\% of the vote and 60\% of the seats in another election, but a metric did not shift in the party's favor.  The principle would also be violated if a party's vote share increased from 55\% to 60\%, its seat share stayed constant at 55\%, and a metric did not register a worsening in the party's position.
\end{quotation}

First we note that, given the examples that  Stephanopoulos and McGhee stated, their intended definition of efficiency principle should instead read:

\begin{quotation}
The efficiency principle states that a measure of partisan gerrymandering must indicate \emph{both} 
\begin{enumerate}
\item[\bf EP1] a greater advantage for (against) a party when the seat share for that party increases (decreases) without any corresponding increase (decrease) in its vote share 
\item[\bf EP2] \emph{and a greater advantage against (for) a party when the vote share increases (decreases) without any corresponding increase (decrease) in the seat share}.
\end{enumerate}
\end{quotation}

The efficiency principle can be stated more mathematically as follows: consider election data consisting of a districting plan $\D$ and geographic distribution of votes $\Delta$.  Suppose $G(\D, \Delta)$ is a function intended to measure partisan gerrymandering.   Given $\D, \Delta$, $G$, and choice of party $A$, one can calculate
\begin{align*}
S(\D,\Delta) &= \text{seat share for party } A \\
V(\Delta) &= \text{vote share for party } A \\
G(\D, \Delta) &= \text{number such that a \emph{higher} value means the districting plan is \emph{more favorable}} \\
&  \quad   \text{ to party } A 
\end{align*}
Now the efficiency principle states:
\begin{equation*}
\text{EP1:   If } V(\Delta)=V(\Delta') \text{ and } S(\D, \Delta)<S(\D', \Delta') \text{ then we must have } G(\D, \Delta)<G(\D', \Delta')
\end{equation*}
\begin{equation*}
\text{EP2:   If } V(\Delta)<V(\Delta') \text{ and } S(\D, \Delta)=S(\D', \Delta') \text{ then we must have } G(\D, \Delta)>G(\D', \Delta')
\end{equation*}

Several political scientists and legal scholars believe that the efficiency principle is a key requirement of a metric to measure partisan gerrymandering.     In \cite{MeasureMetricDebate}, Stephanopoulos and McGhee listed consistency with the efficiency principle as the first of four criteria desired of any metric intended to measure partisan gerrymandering.   They also point out in \cite{MeasureMetricDebate} that McDonald and Best \cite{UnfairPartisanGerrymanders} and Gelman and King \cite{EnhancingDemocracy} also have made statements supporting the importance of the efficiency principle.   We note here that most of these scholars seem to understand the efficiency principle as being what is listed as EP1 above.  We added EP2 purely based on the preceding quote from Stephanopoulos and McGhee in \cite{MeasureMetricDebate}.

Cover critiques the efficiency principle in \cite{QuantifyingGerrymandering} and suggests what he calls the \emph{Modified Efficiency Principle}.  The Modified Efficiency Principle insists that EP1 must hold \emph{unless} the expected seat share changes under plausible variation in the vote share.

We do not take a position on the necessity of the efficiency principle, though we do note that requiring it of a tool $G$ may necessitate introducing bizarre shapes in order to keep the value of $G$ low.  Indeed, suppose that G were a score satisfying the efficiency principle, and consider two different states, both receiving 60\% of the vote for party $A$, but with one state having a very uniform distribution while the other is very clustered.  If the districting practices require keeping $G$ below some threshold, then these states will be required to find districting plans $\mathcal{D},\mathcal{D'}$ \emph{which produce similar shares of representation}.  This may necessitate completely different, and possibly visually offensive, shapes.

In Section \ref{EGnoEPSection} we will show that the efficiency gap fails both  EP1 as well as Cover's Modified Efficiency Principle.  In Section \ref{EGnoPart2EP} we show that the efficiency gap fails EP2.

\section{The Efficiency Gap}\label{EGSection}

Here we study the efficiency gap, which is a metric Stephanopoulos and McGhee introduced to measure partisan gerrymandering \cite{PartisanGerrymanderingEfficiencyGap} and the metric we focus on in this article.  This metric was used to argue the presence of partisan gerrymandering in the Wisconsin court case Whitford v. Gill \cite{WhitfordGill}  (later appealed to the Supreme Court), and has been the subject of much debate since.  (See, for example, \cite{MR3699778, MeasureMetricDebate, HowCompetitiveFairSingleMember, FlawsEfficiencyGap}).  As we shall shortly see, EG does not satisfy EP1, EP2, or the less stringent Modified Efficiency Principle.

The examples we will give involve uneven turnout between districts.  All authors have noted that for equal turnout the efficiency gap simplifies to the seat margin minus twice the vote margin, and thus satisfies the efficiency principle (see, for example, the survey \cite{TappPreprint}).  Chambers, Miller, and Sobel seem to be aware that uneven turnout can break down the efficiency principle for the efficiency gap; in a footnote in \cite{FlawsEfficiencyGap} they state ``the efficiency gap may fail to satisfy McGhee's `efficiency principle' when districts do not have equal numbers of voters.''  Cover also highlights the importance of turnout's effect on the efficiency gap in his discussion of the ``turnout gap'' in \cite{QuantifyingGerrymandering}.   We note that McGhee also states in \cite{MeasuringEfficiency} that ``This [effect of turnout] is	 the opposite of what would be expected, and a clear violation of the  EP [efficiency principle].''\footnote{In the same paper, McGhee suggests redefining $EG = S^*-2V^*$ where $S^*$ is the seat margin and $V^*$ is the vote margin.  As mentioned above, this is what EG simplifies to when turnout in all districts is equal.  This formula does nothing more or less than prescribing the seat share as a linear formula in the vote share, and thus we believe the courts will find it unsatisfying.}

In this section we will give explicit examples showing that the efficiency gap does not satisfy the efficiency principle and we will show precisely how this failure to satisfy the efficiency principle is related to turnout factors.  But before we can do this, we must first understand how to compute the efficiency gap.  Throughout this section, we will assume there are only two parties and will use $V$ to denote party $A$'s vote share (its number of votes as a proportion of the total) and $S$ to denote party $A$'s seat share (its number of seats won as a proportion of the total) in a given election.

\subsection{Calculating the Efficiency Gap}\label{CalculatingEGSection}

The efficiency gap is based on the concept of a \emph{wasted vote}.  There are two kinds of wasted votes: the losing vote and the surplus vote.  I've made a losing vote for candidate $A$ if I voted for candidate $A$ but candidate $B$ won my district.  And I've made a surplus vote if already a majority of the population in my district voted for candidate $A$, and I made yet another vote for candidate $A$ on top of that.  Both the losing vote and the surplus vote don't help my candidate get elected, so in either case my vote is wasted.  The efficiency gap subtracts the number of party $A$'s wasted votes from the number of party $B$'s wasted votes, and then divides by the total number of votes. 

More explicitly, suppose a state has $n$ districts and let $V_i^P$ be the number of votes for party $P \in \{A, B\}$ in district $i \in \{1, 2, \dots, n\}$.  Suppose that party $A$ won districts $1, 2, \dots, m$ and lost districts $m+1, m+2, \dots, n$.  Then the efficiency gap is:
\begin{align}
EG &= \frac{\sum_{i=1}^m \left(V_i^B - \left(V_i^A -\frac{V_i^A+V_i^B}{2}\right)\right) + \sum_{j=m+1}^n \left(\left(V_j^B-\frac{V_j^B+V_j^A}{2}\right) -V_j^A\right)}{\sum_{\ell=1}^n\left(V_\ell^A+V_\ell^B\right)} \notag \\
&= \frac{\sum_{i=1}^m\left( \frac{3}{2}V_i^B-\frac{1}{2}V_i^A\right) + \sum_{j=m+1}^n \left(\frac{1}{2}V_j^B - \frac{3}{2}V_j^A\right)}{\sum_{\ell=1}^n\left(V_\ell^A+V_\ell^B\right)}   \label{EG_ThreeTimes}
\end{align}

\begin{Remark}\label{ThreeToOneRemark}
The coefficients of $\frac{3}{2}$  and $\frac{1}{2}$ in equation \eqref{EG_ThreeTimes} in effect cause the efficiency gap to count losing votes three times as much as winning votes.
\end{Remark}

We note that Cover \cite{QuantifyingGerrymandering} and Nagle \cite{HowCompetitiveFairSingleMember} have commented on the fact that the ``surplus vote'' could also be calculated as simply the difference between the number of winning votes and the number of losing votes (as opposed to half of that difference).  This would change the above remark to say that losing votes would be counted twice as much as winning votes.\footnote{\label{ExtremeFootnote} We note that this new definition of surplus votes would have similar outcomes for the rest of this article.  The efficiency gap with this new definition of surplus votes would still not satisfy the efficiency principle.  Theorem \ref{OneFourthFact} would be true with now $1/3<V<2/3$, and Theorem \ref{TurnoutFact} would have a different, somewhat more extreme, ratio of turnout factors: $\frac{S(3V-2)}{(S-1)(3V-1)}$.} Interestingly, the calculation of surplus votes seems to be part of Judge Griesbach's critique of the efficiency gap.  In his dissenting opinion in Whitford v. Gill, he states 
\begin{quotation}
For example, if the Indians defeat the Cubs 8 to 2, any fan might say that the Indians ``wasted'' 5 runs, because they only needed 3 to win yet scored 8. Under the Plaintiff's theory, however, the Indians needed 5 runs to beat the Cubs that day: 4 runs to reach 50\% of the total runs, plus one to win. That, of course, is absurd.
\end{quotation}

Our constructions exploit the observations in Remark \ref{ThreeToOneRemark}.

\subsection{The Efficiency Gap Does Not Satisfy EP1}\label{EGnoEPSection}

Through examples, we shall show the following:
\begin{Proposition}
The efficiency gap does not satisfy EP1.  That is, there exist two election outcomes in which parties $A$ and $B$ receive the same proportion of the vote in both elections, party $A$ wins more districts in the second election, and yet the efficiency gap of both elections is the same.
\end{Proposition}

We prove this Proposition with the sample elections in Table \ref{RealisticExample}.\footnote{We restrict our attention to examples with EG = 0, which suffices to prove our results.  Examples can be constructed with arbitrary fixed EG.}  

\begin{table}[h]
\begin{tabular}{|c||C{1cm}|C{1cm}|C{1cm}|C{1cm}|c||C{1cm}|C{1cm}|C{1cm}|C{1cm}|c|}\hline
& \multicolumn{5}{|c||}{Election 1} & \multicolumn{5}{c|}{Election 2} \\ \hline
District & \multicolumn{2}{c|}{Votes} & \multicolumn{2}{c|}{Wasted Votes} & Turnout & \multicolumn{2}{c|}{Votes} & \multicolumn{2}{c|}{Wasted Votes} & Turnout  \\ \hline
 & A & B & A & B & & A & B & A & B &  \\ \hline
1 & 48 & 52 & 48 & 2 & 100 & 72 & 78 & 72 & 3 & 150 \\ \hline
2 & 48 & 52 & 48 & 2 & 100 & 72 & 78 & 72 & 3 & 150 \\ \hline
3 & 48 & 52 & 48 & 2 & 100 & 72 & 78 & 72 & 3 & 150 \\ \hline
4 & 48 & 52 & 48 & 2 & 100 & 72 & 78 & 72 & 3 & 150 \\ \hline
5 & 48 & 52 & 48 & 2 & 100 & 52 & 48 & 2 & 48 & 100 \\ \hline
6 & 52 & 48 & 2 & 48 & 100 & 52 & 48 & 2 & 48 & 100  \\ \hline
7 & 52 & 48 & 2 & 48 & 100 & 52 & 48 & 2 & 48 & 100 \\ \hline
8 & 52 & 48 & 2 & 48 & 100 & 52 & 48 & 2 & 48 & 100 \\ \hline
9 & 52 & 48 & 2 & 48 & 100 & 52 & 48 & 2 & 48 & 100 \\ \hline
10 & 52 & 48 & 2 & 48 & 100 & 52 & 48 & 2 & 48 & 100 \\ \hline\hline
Total & 500 & 500 & 250 & 250 & 1000 & 600 & 600 & 300 & 300 &1200  \\ \hline
$V$ & \multicolumn{5}{|c||}{500/1000=50\%} & \multicolumn{5}{c|}{ 600/1200 = 50\%} \\ \hline
$S$ & \multicolumn{5}{|c||}{5/10=50\%} & \multicolumn{5}{c|}{ 6/10 = 60\%} \\ \hline
EG & \multicolumn{5}{|c||}{(250-250)/1000=0} & \multicolumn{5}{c|}{ (300-300)/1200=0} \\ \hline
\end{tabular}
\caption{The EG does not satisfy EP1}
\label{RealisticExample}
\end{table}

We can see that the average turnout in the districts that party $A$ loses is 1.5 times higher than the average turnout in the districts that party $A$ wins.  It is this lopsided turnout that allows the efficiency gap to be 0:  party $A$ has significantly more losing votes in the districts it lost than party $B$ does in the districts that it lost.

One might suggest that a factor of 1.5 in turnout difference is alarming, but we note that this particular number (1.5) is not different from what is seen in practice.  For example,  Tables \ref{TexasTurnoutsRep} and \ref{TexasTurnoutsDem} give the turnouts in the 2016 congressional election in Texas.  For simplicity, we count all votes in each district (even though some were for 3rd party candidates).

\begin{table}[h]
\centering
\begin{tabular}{|c||c|c|c|c|c|c|c|}\hline
District & 1 & 2 & 3 & 4 & 5 & 6 & 7  \\ \hline
Turnout & 260,409 &278,236 & 316,467 & 246,220 & 192,875 & 273,296 & 255,533 \\ \hline\hline
District &  8 &  10 &  11 & 12& 13 & 14 &  17  \\ \hline
Turnout & 236,379 & 312,600   & 225,548  & 283,115 &    221,242 &  259,685 &245,728  \\ \hline\hline
District &  19  & 21 & 22 &  23 & 24 & 25 &26  \\ \hline
Turnout &  203,475 &  356,031 &  305,543 & 228,965 &  275,635 &  310,196 & 319,080  \\ \hline\hline
District & 27 &  31&  32 &  36   \\ \cline{1-5}
Turnout & 230,580  & 284,588 &  229,171 &  218,565  \\ \cline{1-5}
\end{tabular}
\caption{Republican-won Districts in the 2016 Texas congressional election.}
\label{TexasTurnoutsRep}
\end{table}

\begin{table}[h]
\centering
\begin{tabular}{|c||c|c|c|c|c|c|c|}\hline
District &  9 & 15 &  16 & 18 & 20 & 28 &  29  \\ \hline
Turnout & 188,523 &  177,479 & 175,229 &  204,308 &187,669 &  184,442 & 131,982  \\ \hline\hline
District &  30 &  33 & 34 & 35   \\  \cline{1-5}
Turnout &   218,826 &  126,369 &  166,961 & 197,576  \\  \cline{1-5}
\end{tabular}
\caption{Democrat-won Districts in the 2016 Texas congressional election.}
\label{TexasTurnoutsDem}
\end{table}

Just glancing at Tables \ref{TexasTurnoutsRep} and \ref{TexasTurnoutsDem} we can see that the Republican-won districts have higher turnouts in general.  If we calculate the turnout ratio, we find that
\begin{equation*}
\frac{\text{average turnout in districts Republicans won}}{\text{average turnout in districts Democrats won}} = \frac{262,766.48}{178,124}\approx 1.475
\end{equation*}
which is quite close to 1.5.  

For comparison, we have data (rounded to the nearest hundredth) from the 2016 U.S. House of Representatives elections in all states with 8 or more congressional districts in  Table \ref{TurnoutTable} \cite{AZTurnout, CATurnout, FLTurnout, GATurnout, ILTurnout, INTurnout, MDTurnout, MATurnout, MITurnout, MNTurnout, MOTurnout, NJTurnout, NYTurnout, NCTurnout, OHTurnout, PATurnout, TNTurnout, TXTurnout, VATurnout, WATurnout, WITurnout}.  In this table, $n$, $\rho$  and $M/m$ are defined as follows:
\begin{align*}
n &= \text{number of districts in the state} \\
\rho &= \frac{\text{average turnout in districts Republicans won in the state}}{\text{average turnout in districts Democrats won in the state}} \\
M/m &= \frac{\text{maximum turnout in a single district in the state}}{\text{minimum turnout in a single district in the state}}
\end{align*}

\begin{table}[h]
\centering
\begin{tabular}{|c||c|c|c|c|c|c|c|c|c|c|c|}\hline
State  &AZ & CA & FL &  GA &  IL & IN & MD & MA & MI & MN & MO\\ \hline
$n$ & 9  &53 & 27 &  14 &18 & 9 & 8 & 9 &14 &8 & 8   \\  \hline
$\rho$ &1.42 & 1.11 &  1.07 & 0.99 &1.14 &1.18 &1.08 & 0 & 1.11 & 1.08 & 1.10  \\ \hline
$M/m$  & 2.15 & 4.41 & 1.62&1.55 & 2.06 &1.42 &1.18 & 1.34 & 1.47 &  1.19 & 1.34  \\    \hline \hline 
State  & NJ & NY &  NC & OH & PA &  TN & TX & VA & WA & WI  \\ \cline{1-11}
$n$ & 12 & 27 &  13 & 16 & 18 & 9 & 36 & 11 & 10& 8  \\  \cline{1-11}
$\rho$ & 1.26 & 1.11& 0.94 &1.14 & 1.00 & 1.10 &  1.48 &1.09 & 0.91 & 1.16  \\ \cline{1-11}
$M/m$ &  1.96 & 1.83 & 1.27 & 1.34 & 1.56 & 1.31 & 2.82 & 1.42 & 1.65 & 1.53  \\ \cline{1-11}
\end{tabular}
\caption{Turnout ratios in all states with at least 8 congressional districts}
\label{TurnoutTable}
\end{table}

It is important to note that some of this information may have confounding variables.  For example, in California  and Washington the primary process allows for two candidates of the same party to be the two candidates in the final election (which may, in turn, affect turnout).  Additionally, many of the states had third-party or write-in candidates.  For Table \ref{TurnoutTable}, we used the \emph{total turnout} in our calculations (including third-party or write-in candidates).  Note that nearly all of the ratios $\rho$ are larger than 1 (Massachusetts is unusual in that it had no congressional seats won by Republicans).

We note that the  example in Table \ref{RealisticExample} involved a state with a larger number of districts: 10.  In general, as we shall see in section \ref{TurnoutSection}, pairs of elections can realistically be constructed to fail EP1 when a state has a larger number of districts.  This suggests that the efficiency gap has concerning flaws when used for states with a larger number of districts.  In \cite{MR3699778}, Bernstein and Duchin also discuss the fact that the efficiency gap's lack of ``granularity'' make it problematic for states with few districts, and Stephanopoulos and McGhee themselves only compiled historical data in their original paper from races with at least 8 seats \cite{PartisanGerrymanderingEfficiencyGap}.

Finally, in Table \ref{CoverExample}, we give another example to show that the efficiency gap does not satisfy Cover's ``modified efficiency principle'' stated at the end of Section \ref{EPSection}.  

\begin{table}[h]
\begin{tabular}{|c||C{1cm}|C{1cm}|C{1cm}|C{1cm}|c||C{1cm}|C{1cm}|C{1cm}|C{1cm}|c|}\hline
& \multicolumn{5}{|c||}{Election 1} & \multicolumn{5}{c|}{Election 2} \\ \hline
District & \multicolumn{2}{c|}{Votes} & \multicolumn{2}{c|}{Wasted Votes} & Turnout & \multicolumn{2}{c|}{Votes} & \multicolumn{2}{c|}{Wasted Votes} & Turnout  \\ \hline
 & A & B & A & B & & A & B & A & B &  \\ \hline
1 & 25 & 75 & 25 & 25 & 100 &37 & 113 & 37 & 38 & 150 \\ \hline
2 & 25 & 75 & 25 & 25 & 100 &37 & 113 & 37 & 38 & 150 \\ \hline
3 & 25 & 75 & 25 & 25 & 100 &38 & 112 & 38 & 37 & 150 \\ \hline
4 & 25 & 75 & 25 & 25 & 100 &38 & 112 & 38 & 37 & 150 \\ \hline
5 & 25 & 75 & 25 & 25 & 100 &75 & 25 &25 & 25 & 100 \\ \hline
6 & 75 & 25 & 25 & 25 & 100 &75 & 25 &25 & 25 & 100 \\ \hline
7 & 75 & 25 & 25 & 25 & 100 &75 & 25 &25 & 25 & 100 \\ \hline
8& 75& 25 & 25 & 25 & 100 &75 & 25 &25 & 25 & 100 \\ \hline
9 & 75 & 25 & 25 & 25 & 100 &75 & 25 &25 & 25 & 100 \\ \hline
10 & 75 & 25 & 25 & 25 & 100 &75 & 25 &25 & 25 & 100 \\ \hline\hline
Total & 500 & 500 & 250 & 250 & 1000 & 600 & 600 & 300 & 300 &1200  \\ \hline
$V$ & \multicolumn{5}{|c||}{500/1000=50\%} & \multicolumn{5}{c|}{ 600/1200 = 50\%} \\ \hline
$S$ & \multicolumn{5}{|c||}{5/10=50\%} & \multicolumn{5}{c|}{ 6/10 = 60\%} \\ \hline
EG & \multicolumn{5}{|c||}{(250-250)/1000=0} & \multicolumn{5}{c|}{ (300-300)/1200=0} \\ \hline
\end{tabular}
\caption{An Example Where Plausible Variation in Vote Share Does Not Affect Expected Seat Share}
\label{CoverExample}
\end{table}
Tables \ref{RealisticExample} and \ref{CoverExample} show that both close races and landslides can happen in pairs of elections having the same vote share, same efficiency gap, but different seat share.  Since small variation in vote share would not affect the expected seat share in elections from Table \ref{CoverExample}, we see that the efficiency gap does not satisfy Cover's modified efficiency principle.

\subsection{The Efficiency Gap Does Not Satisfy  EP2}\label{EGnoPart2EP}

\begin{Proposition}
The efficiency gap does not satisfy EP2.  That is, there exist two elections in which party $A$ wins the same number of seats, party $A$ wins a higher percentage of the total vote in the second election, and yet the efficiency gaps of the two elections are both equal.
\end{Proposition}

This fact is proven by the examples in Table \ref{Part2Example}.

\begin{table}[h]
\begin{tabular}{|c||C{1cm}|C{1cm}|C{1cm}|C{1cm}|c||C{1cm}|C{1cm}|C{1cm}|C{1cm}|c|}\hline
& \multicolumn{5}{|c||}{Election 1} & \multicolumn{5}{c|}{Election 2} \\ \hline
District & \multicolumn{2}{c|}{Votes} & \multicolumn{2}{c|}{Wasted Votes} & Turnout & \multicolumn{2}{c|}{Votes} & \multicolumn{2}{c|}{Wasted Votes} & Turnout  \\ \hline
 & A & B & A & B & & A & B & A & B &  \\ \hline
1 & 70 & 30 & 20 & 30 & 100 & 75 & 25 & 25 & 25 & 100 \\ \hline
2 & 60 & 10 & 25 & 10 & 70 & 75 & 25 & 25 & 25 & 100 \\ \hline
3 & 30 & 100 & 30 & 35 & 130 & 25 & 75 & 25 & 25 & 100\\ \hline\hline
Totals & 160 & 140 & 75 & 75 & 300 & 175 & 125 & 75 & 75 & 300 \\ \hline
V & \multicolumn{5}{|c||}{$160/300 \approx 53\%$ } &      \multicolumn{5}{c|}{ $175/300 \approx 58\% $} \\ \hline
S & \multicolumn{5}{|c||}{$2/3 \approx 67\%$} & \multicolumn{5}{c|}{ $2/3 \approx 67\%$} \\ \hline
EG & \multicolumn{5}{|c||}{(75-75)/300=0} & \multicolumn{5}{c|}{ (75-75)/300=0} \\ \hline
\end{tabular}
\caption{Party $A$ Wins the Same Districts in Both Elections,   about 53\% of Votes in Election 1, about 58\% of Votes in Election 2, and EG=0 in Both Elections}
\label{Part2Example}
\end{table}

\subsection{The 1/4 Boundary for ``Hidden'' Gerrymandering}\label{OneFourthSection}

We previously saw an example showing that the efficiency gap fails EP1.  Here, we shall prove the following:

\begin{Theorem}\label{OneFourthFact}
For any rational numbers $1/4<V<3/4$ and $0<S<1$, there exists election data with vote share $V$, seat share  $S$, and  $EG = 0$.
\end{Theorem}

On the face of it, this Theorem is startling.  It states that there exists an election with $V = 26\%, S = 99\%$, and yet the efficiency gap is 0.  We shall prove Theorem \ref{OneFourthFact} by constructing such an election outcome.  Not surprisingly, this construction will produce quite lopsided turnouts.  Information on how lopsided the turnouts  must be  will be presented in Section \ref{TurnoutSection}.

\begin{proof}[Proof of Theorem \ref{OneFourthFact}]
Fix $S=\frac{m}{n}$ and $V$ and choose $M$ large enough so that $MV\geq m$ and  $M(1-V) \geq n-m$.  First, construct an election outcome in which party $A$ has vote share $V$ and seat share $S$ by  distributing $MV$ votes to party $A$ and no votes to party $B$ in the $m$ districts that party $A$ wins, and $M(1-V)$ votes to party $B$ and no votes to party $A$ in the $n-m$ districts that party $A$ loses.  

Now we adjust this election in a way so as to keep the vote share the same but make the efficiency gap go to 0.  

\textbf{Case 1: $1/4<V \leq1/2$.}  Note that before any adjustment has been made, party $A$ has $MV/2$ wasted votes and party $B$  has $M(1-V)2$ wasted votes.  Note that $M(1-V)/2\geq MV/2$.

We  adjust the election by distributing $NV$ votes to party $A$ and $N(1-V)$ votes to party $B$ in only the districts where party $A$ is losing, where $N$ is a positive number.  The remaining election still has vote proportion $V$ and seat proportion $S$.  But now, since $1/4<V$, the total count of wasted votes for party $A$ has gone up \emph{more} than the total count of wasted votes for party $B$.  More specifically, we know that party $A$'s wasted vote count has gone up by $NV$, while party $B$'s wasted vote count has gone up by $\frac{1}{2}N(1-V) - \frac{1}{2}NV$.  The net count of wasted votes for party $A$ increases so long as
\begin{align*}
NV &> \frac{1}{2}N(1-V)-\frac{1}{2} NV \\
V&>\frac{1}{4}
\end{align*}
 Now, by choosing $N$ appropriately, we can make the total number of wasted votes for party $A$ to be the same as the total number of wasted votes for party $B$.  Thus, we can get the efficiency gap to go to 0.

\textbf{Case 2: $1/2 \leq V < 3/4$}  In this case, we can simply exchange the roles of party $A$ and party $B$, and we are back in Case 1.  
\end{proof}

Let's see how this construction plays out in an extreme case.  Specifically, the case of party $A$ receiving 27\% of the vote and winning 9 out of 10 congressional seats.  If one follows this construction, the resulting outcome is in Table \ref{27Raw}.

\begin{table}[h]
\centering
\begin{tabular}{|c||C{1cm}|C{1cm}|C{1cm}|C{1cm}|c|}\hline
District & \multicolumn{2}{c|}{Votes} & \multicolumn{2}{c|}{Wasted Votes} & Turnout   \\ \hline
 & A & B & A & B &   \\ \hline
 1 & 155 & 492 & 155& 168.5 & 647 \\ \hline
 2 & 3 & 0 & 1.5 & 0 & 3 \\ \hline
 3 & 3 & 0 & 1.5 & 0 & 3 \\ \hline
4 & 3 & 0 & 1.5 & 0 & 3 \\ \hline
5 &3 & 0 & 1.5 & 0 & 3 \\ \hline
6 & 3 & 0 & 1.5 & 0 & 3 \\ \hline
7 & 3 & 0 & 1.5 & 0 & 3 \\ \hline
8& 3 & 0 & 1.5 & 0 & 3 \\ \hline
9 & 3 & 0 & 1.5& 0 & 3 \\ \hline
10 & 3 & 0 & 1.5 & 0 & 3 \\ \hline\hline
Total & 182 & 492 &168.5  &168.5  & 674 \\ \hline
$V$ & \multicolumn{5}{c|}{182/674=27\%} \\ \hline
$S$ & \multicolumn{5}{c|}{9/10=90\%} \\ \hline
EG & \multicolumn{5}{c|}{(168.5-168.5)/674=0} \\ \hline
\end{tabular}
\caption{$V = 27\%$, $S = 90\%$, and EG = 0}
\label{27Raw}
\end{table}

Obviously, the difference in turnout is extreme, but the fact that it is even possible to construct an example such as the one in Table \ref{27Raw} is notable.

\subsection{Turnout Factors}\label{TurnoutSection}

In Election 2 in Tables \ref{RealisticExample} and \ref{CoverExample}, we saw election outcomes having an efficiency gap of 0 where $A$ received half the votes but more than half the seats.  And in  the election from Table \ref{27Raw}, we say party $A$ winning 9/10 of the seats with 27\% of the vote and efficiency gap of 0.  For each of these examples,  the turnout was higher in districts where $A$ lost than in districts where $A$ won.    Interestingly, the mathematics of election outcomes with 0 efficiency gap tells us \emph{precisely} how lopsided such elections will be.  Define turnout ratio $\rho$ as follows:
\begin{equation*}
\rho = \frac{\text{average turnout in districts party } A \text{ lost}}{\text{average turnout in districts party } A \text{ won}}
\end{equation*}
We have the following.

\begin{Theorem}\label{TurnoutFact}
Fix  rational numbers $1/4 < V  <3/4$ and $0<S<1$.   Consider an election with vote share $V$, seat share $S$, and EG=0.  (We know that such an election exists from Theorem \ref{OneFourthFact}).  Then 
\begin{equation*}
\rho = \frac{S(3-4V)}{(1-S)(4V-1)}
\end{equation*}

\end{Theorem}

Upon inspection, this is a mathematically intuitive result in that $\rho$ goes to 0 if either $V$ goes to 3/4 or $S$ goes to 0, and $\rho$ goes to infinity if either $V$ goes to $1/4$ or $S$ goes to 1.  However, it may be surprising that this turnout ratio is an expression \emph{only} in $S$ and $V$.

Before proving Theorem \ref{TurnoutFact}, we must first prove that the set of election outcomes under consideration is a convex cone.  For the reader not familiar with vector spaces, see, for example, \cite{AntonLA}, and for the reader not familiar with convex cones, see, for example, chapter II of \cite{MR1940576}.  

Suppose the number of districts is $n$ and fix an integer $m$ with $1 \leq m \leq n-1$.  We'd like to look at the set $E$ of all election outcomes where party $A$ wins $m$ districts and EG=0 (thus here we have $S = \frac{m}{n}$).  Without loss of generality, we can assume that party $A$ wins the districts labeled $1, 2, \dots, m$ and loses the districts labeled $m+1, m+2, \dots, n$.  Each election outcome can be written as a vector $(a_1, a_2, \dots, a_n, b_1, b_2, \dots, b_n)$ where $a_i$ gives the number of votes for party $A$ in district $i$ and $b_i$ gives the number of votes for party $B$ in district $i$.  Then we can see that $E$ is the set of all vectors in $\Z^{2n}$ satisfying the following:
\begin{align*}
a_i &\geq 0 \quad \quad i=1, 2, \dots,  n \\
b_i &\geq 0 \quad \quad i=1, 2, \dots, n \\
a_i &\geq b_i \quad \quad i=1, 2, \dots, m \\
b_i &\geq  a_i \quad \quad i=m+1, m+2, \dots, n \\
& \sum_{i=1}^m b_i  + \sum_{i=m+1}^n \left(b_i-\frac{a_i+b_i}{2}\right) =  \sum_{i=m+1}^n a_i + \sum_{i=1}^m \left( a_i-\frac{a_i+b_i}{2} \right)
\end{align*}
Note that the inequalities state that the number of votes is non-negative and that party $A$ wins districts $1, 2, \dots, m$ and loses districts $m+1, m+2, \dots, n.$\footnote{Since we state the inequality $a_i \geq b_i$ for $i=1, 2, \dots, m$  not as \emph{strict} inequalities, we view an outcome of $a_i=b_i$ as being the case where both parties get the same number of votes but after drawing a name out of a bowl (or whatever other electoral procedure) party $A$ won the district. We similarly interpret the cases where $b_j=a_j$, $j=m+1, m+2, \dots, n$. }  The equality states that EG=0.  

For ease, we need to introduce some notation.  Let $\alpha_i$ denote the vector indicating a vote of 1 for party $A$ in district $i$ and votes of 0 everywhere else.  Let $\beta_i$ denote the vector indicating a vote of 1 for party $B$ in district $i$ and votes of 0 everywhere else.  That is, if $e_i\in \R^{2n}$ denotes the $i$th standard basis vector, then 
\begin{align*}
\alpha_i &= e_i \\
\beta_i &= e_{n+i}
\end{align*}

We shall first prove the following:

\begin{Claim}\label{ConeClaim}
Every vector in $E$ can written as a non-negative integer combination of vectors of one of the following forms:
\begin{align}\label{one}
 & \alpha_i+\alpha_j+\alpha_k+\beta_i & \quad \quad \text{ for } i, j, k \in \{1, 2, \dots, m\} \\ \label{two}
& \beta_i+\beta_j+\beta_k+\alpha_i &\quad \quad \text{ for } i, j, k \in \{m+1, m+2, \dots, n\} \\ \label{three}
 &\alpha_i+\beta_j &\quad \quad \text{ for } i \in \{1, 2, \dots, m\}, j \in \{m+1, m+2, \dots, n\} \\ \label{four}
& \alpha_i+\beta_i+\alpha_j+\beta_j & \quad \quad\text{ for } i \in \{1, 2, \dots, m\}, j \in \{m+1, m+2, \dots, n\} 
\end{align}

\end{Claim}

The reader will notice that we exploit Remark \ref{ThreeToOneRemark} in our proof.

\begin{proof}[Proof of Claim \ref{ConeClaim}]
We proceed with proof by contradiction.  Suppose our Claim were false, and there existed a vector in $E$ which could not be written as a non-negative integer combination of vectors of type \eqref{one}, \eqref{two}, \eqref{three}, and \eqref{four}.  Since all of the vectors in $E$ have positive integer coordinates, we can find one such vector such that the sum of that vector's coordinates is minimized.  That is, $\epsilon = (a_1, a_2, \dots, a_n, b_1, b_2, \dots, b_n) \in E \subset \Z^{2n}$ is chosen such that
\begin{equation*}
\sum_{i=1}^n\left(a_i+b_i\right)
\end{equation*}
is minimized.  Certainly $\epsilon \not= \vec{0}$ because $\vec{0}$ is a non-negative integer combination of the above vectors (the linear combination with all coefficients 0).  Thus, $\epsilon$ has a positive coordinate.  

\textbf{Case 1: $\epsilon$ has positive coordinates corresponding to losing votes, but only for party $B$.}   In this case, the only place where party $A$ can have wasted votes is with surplus votes.  Party $A$ must firstly win the district in which $B$ has a losing vote, and since EG=0 for election $\epsilon$,  party $A$ must have at least 2 surplus votes.  Thus, we can find $i, j, k \in \{1, 2, \dots, m\} $ for which
\begin{equation*}
\epsilon - \left( \alpha_i+\alpha_j+\alpha_k+\beta_i \right) \in E
\end{equation*}
contradicting that $\epsilon$ was the element in $E$ with smallest sum of coordinates which could not be written as a non-negative integer combination of the above vectors.

\textbf{Case 2: $\epsilon$ has positive coordinates corresponding to losing votes, but only for party $A$.}  By the same argument as in case 1, we find a contradiction.

\textbf{Case 3: both parties $A$ and $B$ have a losing vote.}  In this case, we can find an $i \in \{1, 2, \dots, m\}, j \in \{m+1, m+2, \dots, n\} $ such that
\begin{equation*}
\epsilon-\left( \alpha_i+\beta_i+\alpha_j+\beta_j \right) \in E
\end{equation*}
again contradicting that $\epsilon$ was the element in $E$ with smallest sum of coordinates which could not be written as a non-negative integer combination of the above vectors.

\textbf{Case 4: $\epsilon$ does \emph{not} have any positive coordinates which are losing votes.}  Here, $\epsilon$ only has positive coordinates which are winning votes.  Again, since EG=0 for election $\epsilon$, there is some $ i \in \{1, 2, \dots, m\}, j \in \{m+1, m+2, \dots, n\}$ for which
\begin{equation*}
\epsilon - \left(\alpha_i+\beta_j \right) \in E
\end{equation*}
again contradicting that $\epsilon$ was the element in $E$ with smallest sum of coordinates which could not be written as a non-negative integer combination of the above vectors.

We have exhausted all cases, and thus we must conclude that Claim \ref{ConeClaim} is true.
\end{proof}

Note that Claim \ref{ConeClaim} allows us to see an alternate proof of Theorem \ref{OneFourthFact}.  Suppose we want to find an election outcome with EG=0  in which party $A$ has seat share $S = \frac{m}{n}$ and vote share $V$.  By Claim \ref{ConeClaim}, we need only form a non-negative integer combination of the vectors in equations \eqref{one}, \eqref{two}, \eqref{three}, and \eqref{four} with the \emph{additional} constraint that 
\begin{equation*}
\sum_{i=1}^n a_i = V\left(\sum_{i=1}^n(a_i+b_i)\right)
\end{equation*}
Upon inspection of the proportion of votes given to each party in equations \eqref{one}, \eqref{two}, \eqref{three}, and \eqref{four}, this can be achieved nontrivially (meaning without any districts with a tie between parties $A$ and $B$) so long as $1/4<V<3/4$.  

With Claim \ref{ConeClaim} in hand, we can now prove Theorem \ref{TurnoutFact}.

\begin{proof}[Proof of Theorem \ref{TurnoutFact}]
Consider an election in which party $A$ has vote share $1/4<V<3/4$ and seat share $S = \frac{m}{n}$ (winning $m$ seats) and  EG=0.   From Claim \ref{ConeClaim} we know that this election is a non-negative integer combination of vectors of types \eqref{one}, \eqref{two}, \eqref{three}, and \eqref{four}.

%

\textbf{Case 1:  $1/4<V\leq1/2$.}  First  we suppose the election under consideration is a non-negative integer combination only of vectors of types \eqref{one} and \eqref{two}.  Let $k_{\ref{one}}$ be the sum of the coefficients  of vectors of type \eqref{one} and $k_{\ref{two}}$ be the sum of the coefficients of vectors of type \eqref{two} in that conic combination.  Then, to satisfy these constraints, we must have
\begin{align*}
3k_{\ref{one}}+k_{\ref{two}} &= TV \\
k_{\ref{one}}+3k_{\ref{two}} &= T(1-V)
\end{align*}
where $T$ is the total turnout.  Solving this pair of linear inequalities for $k_{\ref{one}}$ and $k_{\ref{two}}$ gives
\begin{align*}
k_{\ref{one}} &= \frac{1}{8}T(4V-1) \\
k_{\ref{two}}&= \frac{1}{8}T(3-4V)
\end{align*}
With these coefficients, the districts where party $A$ wins will get a total of $4k_{\ref{one}}$ votes and the districts where party $A$ loses will get a total of $4k_{\ref{two}}$ votes.  Thus we calculate: 
\begin{equation*}
\rho = \frac{\frac{4k_{\ref{two}}}{(n-m)}}{\frac{4k_{\ref{one}}}{m}} = \frac{ m\frac{1}{2}T(3-4V)}{(n-m)\frac{1}{2}T(4V-1)} = \frac{S(3-4V)}{(1-S)(4V-1)}
\end{equation*}
which is the statement of Theorem \ref{TurnoutFact}.  

Now suppose instead our election is a non-negative integer combination only vectors of type \eqref{two} and \eqref{three}, and  $k_{\ref{two}}$ is the sum of the coefficients of vectors of type \eqref{two} and $k_{\ref{three}}$ is the sum of the coefficients of vectors of type \eqref{three}.   To satisfy our election constraints (again assuming our election has $V$ total votes) we have
\begin{align*}
k_{\ref{two}}+k_{\ref{three}} &= TV \\
3k_{\ref{two}}+k_{\ref{three}} &= T(1-V)
\end{align*}
Solving this pair of linear inequalities for $k_{\ref{two}}$ and $k_{\ref{three}}$ gives
\begin{align*}
k_{\ref{two}} &= \frac{1}{2}T(1-2V) \\
k_{\ref{three}} &= \frac{1}{2}T(4V-1)
\end{align*}
In this case the districts where party $A$ wins will get a total of $k_{\ref{three}}$ votes and the districts where party $A$ loses get a total of $4k_{\ref{two}}+k_{\ref{three}}$ votes.  Thus we calculate: 
\begin{equation*}
\rho = \frac{\frac{4k_{\ref{two}}+k_{\ref{three}}}{n-m}}{\frac{k_{\ref{three}}}{m}} = \frac{m\left(2T(1-2V)+\frac{1}{2}T(4V-1)\right)}{(n-m)\frac{1}{2}T(4V-1)} = \frac{m\left(\frac{3}{2}-2V\right)}{(n-m)(2V-\frac{1}{2})} =\frac{S(3-4V)}{(1-S)(4V-1)}
\end{equation*}
which is again the same.  

 Since $V\leq 1/2$, any non-negative integer combination of vectors of type \eqref{one}, \eqref{two}, and \eqref{three} giving a proportion of votes $V$ to party $A$ can be separated into a sum of vectors of types \eqref{one} and \eqref{two} giving a proportion of votes $V$ to party $A$ and another sum of vectors of types \eqref{two} and \eqref{three} also giving a proportion of votes $V$ to party $A$.  Finally, note that  the distribution of votes among parties $A$ and $B$ is the same using vectors of the form \eqref{three} and \eqref{four}.   Thus, any election with EG=0, vote share $V$ and seat share $S = \frac{m}{n}$ must have turnout ratio $\rho= \frac{S(3-4V)}{(1-S)(4V-1)}$.
 
 \textbf{Case 2:  $1/2<V<3/4$.}  Note that  party $B$'s seat share is $1-S$ and vote share is $1-V$ with $1/4<1-V<1/2$.  Thus, putting party $B$ in the role of party $A$ in Theorem \ref{TurnoutFact}, we can use what we've already proved to get 
\begin{align*}
\rho &=\frac{\text{average turnout in districts party } A \text{ lost}}{\text{average turnout in districts party } A \text{ won}}  \\
&=  \left( \frac{\text{average turnout in districts party } B \text{ lost}}{\text{average turnout in districts party } B \text{ won}}\right)^{-1} \\
&= \left(  \frac{(1-S)(3-4(1-V))}{(1-(1-S))(4(1-V)-1)}   \right)^{-1} = \frac{S(3-4V)}{(1-S)(4V-1)}
\end{align*}

Thus Theorem \ref{TurnoutFact} is proved.\footnote{We note that Theorem \ref{TurnoutFact} is still true for the boundary cases $V = 1/4$ and $V = 3/4$.  
 If $V=1/4$, this would correspond to an election where all of the districts that party $A$ wins have NO voter turnout (0 votes for both parties) and after pulling a name out of a bowl (or other electoral procedure), the seat is given to party $A$.  The districts where Party $A$ loses have 25\% of the vote to party $A$ and 75\% to party $B$.  In this election, the turnout factor from Theorem \ref{TurnoutFact} would be infinite, which is appropriate given the fact that there is 0 voter turnout in the districts that party $A$ wins.
 
  If $V=3/4$, this would correspond to an election where all of the districts that party $A$ \emph{loses} have no voter turnout and after pulling a name out of a bowl (or other electoral procedure), the seat is given to party $B$.  The districts where Party $A$ wins have 75\% of the vote to party $A$ and 25\% to party $B$.  In this election, the turnout factor from Theorem \ref{TurnoutFact} would be 0, which is again appropriate given the fact that there is 0 voter turnout in the districts which party $A$ lost. }
\end{proof}

\subsection{Implications of Theorem  \ref{TurnoutFact}}\label{FixingTSection}

Recall that we define $ \frac{\text{average turnout in districts } A \text{ lost}}{\text{average turnout in districts } A \text{ won}} =\rho$ and when $1/4<V<3/4$ and EG=0 we have
\begin{equation*}
\rho=  \frac{S(3-4V)}{(1-S)(4V-1)}
\end{equation*}
Solving this equation for $S$ we have
\begin{equation}\label{S_equation}
S = \frac{\rho(4V-1)}{\rho(4V-1)+3-4V}
\end{equation}
This equation gives the EG-preferred seat share when the vote share is $1/4<V<3/4$ and the turnout ratio is $\rho$.  
\begin{Example}
Recall that in Tables \ref{TexasTurnoutsRep} and \ref{TexasTurnoutsDem}, we saw that in the 2016 Texas congressional elections,
\begin{equation*}
\frac{\text{average turnout in districts Republicans won}}{\text{average turnout in districts Democrats won}} = \frac{262,766.48}{178,124}\approx 1.475
\end{equation*}
We can see that the turnout ratio in our Texas example skews the efficiency gap in favor of the Democrats.  Letting Party $A$ be the Democratic party, when the vote share is 50\%, the EG-preferred Democratic seat share for this Texas election is 
\begin{equation*}
S =\frac{\rho(4V-1)}{\rho(4V-1)+3-4V} = \frac{1.475(4(0.5)-1)}{4(0.5)(1.475-1)+3-1.475} \approx 60\%
\end{equation*}
(If the turnout in each district were equal, when the vote share is 50\% the EG-preferred seat share would also be 50\%).

\end{Example}

The upshot of equation \eqref{S_equation} is:  while holding other factors equal, the more strongly that party $A$ obtains lower turnout in districts it wins than in districts it loses, the more strongly EG reports a gerrymander favoring party $B$.  We can also see this with graphs of the EG-preferred $S$ given in equation \eqref{S_equation}.   These can be seen in Figure \ref{TurnoutGraphs}.

\begin{figure}[h]
\centering
\includegraphics[width=4.5in]{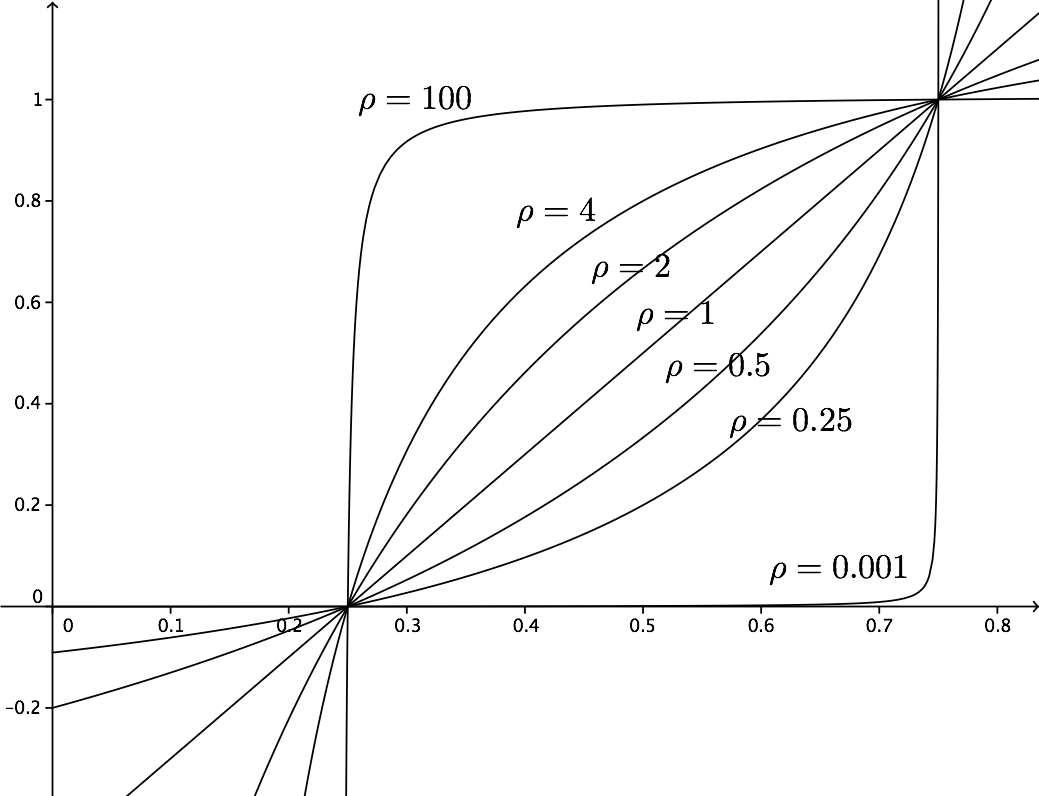}
\caption{Graphs of  EG-preferred $S$ as a function in $V$, for varying values of $\rho$.}
\label{TurnoutGraphs}
\end{figure}

Note that, as $\rho$ increases to infinity, EG=0 requires party $A$ to receives nearly all of the seats no matter their votes.  As $\rho$ decreases to 0, the efficiency gap only thinks that an election outcome is ``perfectly fair'' if party $A$ receives almost no seats.  Also, as expected, the function for $S$ breaks down (gives values less than 0 or larger than 1) when $V<.25$ or $V>.75$.

We can show that the efficiency gap can be calculated only as a function in $S$, $V$, and $\rho$:

\begin{Theorem}\label{EGCalculationTheorem}
Consider an election with seat share $S$, vote share $V$, and turnout ratio $\rho$.  Then the efficiency gap of this election is
\begin{equation*}
EG = S^* -2V^* + \frac{S(1-S)(1-\rho)}{S(1-\rho)+\rho}
\end{equation*}
where $S^*= S-\frac{1}{2}$ is the seat margin and $V^* = V-\frac{1}{2}$ is the vote margin.\footnote{Theorem \ref{EGCalculationTheorem} gives another proof of Theorem \ref{TurnoutFact}, by setting EG=0.  We keep the proof in section \ref{TurnoutSection} because it can be generalized to a different formulation of the efficiency gap, as in footnote \ref{ExtremeFootnote}.}
\end{Theorem}

\begin{proof}
The proof uses Cover's expression of $EG$ using the ``turnout gap'' \cite{QuantifyingGerrymandering}.  Suppose there are $n$ districts with party $A$ winning districts $1, 2, \dots, m$ and party $B$ winning districts $m+1, m+2, \dots, n$.  Let $T_i$ be the turnout in district $i$ and let $T_P$ be the total turnout in districts that party $P$ wins, $P \in \{A,B\}$ so that $T_A = \sum_{i=1}^mT_i$ and $T_B = \sum_{j=m+1}^nT_j$.  In the section on the Turnout Gap in \cite{QuantifyingGerrymandering}, Cover shows:
\begin{equation*}
EG = S^* -2V^*+S\left(\frac{\frac{T_A}{m}}{\frac{T_A+T_B}{n}}-1\right)
\end{equation*}
First suppose that $S \not=1$.  A little algebra (recalling that $S = \frac{m}{n}$, $1-S = \frac{n-m}{n}$, and $\rho = \frac{mT_B}{(n-m)T_A}$) gives:
\begin{align*}
EG &=  S^* -2V^*+S\left(\left(\frac{\frac{T_A+T_B}{n}}{\frac{T_A}{m}}\right)^{-1}-1\right) \\
&=  S^* -2V^*+S\left(\left(S+ \frac{\frac{T_B}{n}}{\frac{T_A}{m}}\right)^{-1}-1\right) \\
&=  S^* -2V^*+\frac{S}{1-S}\left(\left(\frac{S}{1-S}+\rho\right)^{-1}-1+S\right) \\
&= S^* -2V^* +  \frac{S(1-S)(1-\rho)}{S(1-\rho)+\rho}
\end{align*}
Finally, note that if $S=1$ then $T_B=0$ and $m=n$ so that Cover's equation gives $EG = S^*-2V^*$.  Thus the Theorem is proved.
\end{proof}

We can see that $EG = S^* -2V^*$ \emph{only} when $\rho = 1$, $S=0$, or $S=1$.  Recall that $\rho = 1$ when the average turnout in a district that $A$ lost is the same as the average turnout in a district that $A$ won.  Thus, the result that $EG =S^* -2V^*$ exactly when $\rho = 1$ slightly generalizes the fact that $EG=S^*-2V^*$ when turnout in each district is the same\cite{MR3699778}, and is a re-wording of the result that $EG=S^*-2V^*$ exactly when average turnout in districts that party $A$ won is the same as the overall average turnout\cite{QuantifyingGerrymandering}.  Theorem \ref{EGCalculationTheorem} gives the following corollary: 

\begin{Corollary}\label{FixedFact}
Suppose $\rho$ is fixed.  Then the efficiency gap satisfies the efficiency principle.
\end{Corollary}

\begin{proof}
We can see this by calculating that the partial derivative of $EG$ with respect to $S$ is
\begin{equation*}
\frac{\rho}{(S(1-\rho)+\rho)^2}
\end{equation*}
which is positive, showing that the EG satisfies EP1.  The partial derivative of $EG$ with respect to $V$ is -2 (which is negative), showing that it satisfies EP2.
\end{proof}

%

\section{Additional Comments}\label{CommentsSection}

It is well-known that there are various factors affecting voter turnout, including voter ID laws, the availability of conveniences like early voting and vote-by-mail, electoral competitiveness, voter demographics, and even the weather.  It is equally well-known that Hispanic voters have a considerably lower proportion of citizen voting age population (CVAP) per Census population than most other subgroups, and that in recent elections they tend to favor the Democratic party.  This very likely accounts for lower turnout in Democrat-won districts, and correspondingly large values of $\rho$, for Texas and Arizona in particular.

Given that lopsided turnout among different districts within a state is not uncommon (see Table \ref{TurnoutTable}), it is important to know how our tools intended to measure partisan gerrymandering are affected by voter turnout. We have shown that unequal turnout among districts causes EG to ``expect'' an exaggerated seat bonus for the party with lower turnout in the districts that it wins, which as we have seen is currently most often the Democratic party.

The results given here suggest that additional care should be taken when interpreting the numerical values of the efficiency gap.  We strongly caution against using a fixed numerical cutoff such as $|\text{EG}|>.08$ for detecting gerrymanders, and we argue that values of EG should not be compared from one state to another or between different historical periods because of the confounding effects of turnout ratios. 

\section*{Acknowledgments}  
The author would like to express her sincere and deep thanks to M. Duchin for helpful feedback and insightful questions that led to significant improvements.  The author would also like to thank A. Rappaport for bringing the paper \cite{MeasureMetricDebate} to her attention, J. Nagle for helpful preliminary comments, and E. McGhee for his comments.

\bibliographystyle{plain}
\bibliography{EGNotEfficient}

\end{document}